\documentclass[12pt]{iopart}

\usepackage[english]{babel}
\usepackage{amssymb}
\usepackage{harvard}
\usepackage{graphicx}

\newtheorem{theorem}{Theorem}
\newcommand{\theoremend}{\nopagebreak $\qquad \blacklozenge$ \vspace*{4pt}}

\newtheorem{definition}{\vspace*{4pt} Definition \newline}
\newcommand{\definitionend}{\nopagebreak $\qquad \blacklozenge$ \vspace*{4pt}}

\newenvironment{remark}{\vspace*{4pt} \textbf{Remark} $ \quad $} {\nopagebreak \newline \null $\qquad \clubsuit$\vspace*{4pt}}
\newenvironment{proof} {\vspace*{4pt} \textsc{Proof} \begin{quote} \begin{small}} {\end{small} \end{quote} \nopagebreak $\qquad \heartsuit$ \vspace*{4pt}}

\begin{document}

\title[Effective dynamics on thin Dirichlet waveguides]{Effective Schr{\"o}dinger dynamics on $ \varepsilon $-thin Dirichlet waveguides via Quantum Graphs I: star-shaped graphs}

\author{G F Dell'Antonio and E Costa}

\address{Mathematical Physics Sector, SISSA, via Beirut 2-4, 34151 Trieste (ITALY)}

\eads{\mailto{gianfa@sissa.it}, \mailto{ecosta@sissa.it}}

\vspace*{1cm}

\begin{center}
\textit{in memory of Pierre Duclos}
\end{center}

\vspace*{1cm}

\begin{abstract}
We describe the boundary conditions at the vertex that one must choose to obtain a dynamical system that best describes the low-energy part of the evolution of a quantum system confined to a very small neighbourhood of a star-shaped metric graph.
\end{abstract}


\submitto{\JPA}

\maketitle

\section{Introduction}

Often in Physics at the nanoscale one deals with systems which can be regarded as having support in a very small neighborhood $ \Gamma^\varepsilon \subseteq \mathbb{R}^3 $ of a one dimensional graph $\Gamma$. The parameter $ \varepsilon $ is a measure of  the distance of $ \partial \Gamma^\varepsilon $ from $\Gamma$ as compared to the length of a typical edge of the graph.

The dynamics of such systems is described by a propagation equation, which can be the wave equation in the case of optical wires or the Schr{\"o}dinger equation in the case of conduction electrons in a macromolecule or in the case of conducting nano-devices. In the case of optical wires it is reasonable to choose Neumann (reflecting) boundary conditions at the boundary $ \partial \Gamma^\varepsilon $. In the case of macromolecules or nano-devices on the other hand one can consider that the system is confined in the region $ \Gamma^\varepsilon $ by very strong forces at the atomic scale, which provide a potential having the form of very deep and narrow valleys. In this case the use of Dirichlet boundary conditions can be regarded as a reasonable approximation.

The inital value problem is well posed for both boundary conditions, but in general one is not able to give the solution in an explicit form. It is then useful to search for a dynamical system which is explicitly solvable and at the same time provides an approximation to the physical one when $ \varepsilon $ is very small. In this case it is important to give an estimate of the error one makes in this approximation for the observables of interest.

One can expect that the limit dynamical system be described by a wave equation or respectively a Schr{\"o}dinger equation on the graph. But if this is the case, in order to have a well posed problem one must specify boundary conditions at the vertices. So the problem becomes: which are the boundary conditions (if any) that one must choose at the vertices of the (metric) graph, in order to have an evolution equation on the graph with solutions which give an approximation (to order $ \varepsilon ^\alpha $, for some $ \alpha > 0 $) for the expectation values of the relevant physical observables? One may reasonably expect that the answer depends on the shape of $ \Gamma^\varepsilon $ in the vicinity of the vertices.

In the case of Neumann boundary conditions the answer is known \cite{exner05} for the case of the Schr{\"o}dinger equation, at least for initial data with not too large energy. Indeed in this case the initial wavefunction can be chosen smooth (e. g. in $ W^{2, 2} $) uniformly in $ \varepsilon $ and the solution can be restricted to the graph uniformly when $ \varepsilon \to 0 $. As a consequence, one can define a sequence of maps from $ \Gamma^\varepsilon $ to $ \Gamma $, and prove that the trace on $ \Gamma $ of the resolvent of the Schr{\"o}dinger operator on $ \Gamma^\varepsilon $ converges to the resolvent of a Schr{\"o}dinger operator on $ \Gamma $  with Kirchhoff (coupling) vertex conditions (continuity of the solution, zero sum of the outward derivatives). The situation is entirely different in the case of Dirichlet boundary conditions on $ \partial \Gamma^\varepsilon $: the $ W^{1, 2} $-norm of the initial datum increases without bound as $ \varepsilon \to 0 $, and there is no limit trace on $ \Gamma $. For this reason, the question of the existence of a limit flow had not been answered so far \cite{kuchment08}, in spite of the obvious physical and mathematical interest.

In concrete cases of physical devices, models of limit dynamics have been constructed to fit the experimental data. For instance, in the case of a sharply bent conducting device experimental evidence shows that in order to have a good approximation one must use Dirichlet (decoupling) boundary conditions at the bend. On the other hand, the standard treatment of conducting electrons in aromatic molecules (such as graphene or benzene), in which the molecule is represented by a graph, shows that results in fair accordance with experiments are obtained if the limit model is constructed with conditions at the vertices that are of weighted Kirchhoff type. The \textit{phenomenological} values of the weights are different for different molecules. One can consider this as an evidence that the ``right'' boundary conditions depend on the shape of $ \Gamma^\varepsilon $. 

Indeed the images seen at the electronic microscope are different for different molecules (one takes $ \Gamma^\varepsilon $ to be the region in which the density of conducting electrons is appreciably different from zero) but in all cases they have the shape of an annulus around the nucleus, with tunnels in correspondence of the valence bonds. One interprets this structure as due to the combined action of the attraction by the nucleus and the exclusion principle that forbids the conduction electrons to occupy the region of the core electrons.

Here we give an answer to the mathematical question described above. We show that the generator of the limit dynamics on the graph corresponds to a suitable boundary condition at the $i$-th vertex $ V_i $, which depends on the shape of $ \Gamma^\varepsilon $ near $ V_i $ through the spectral properties of a sequence of auxiliary Schr{\"o}dinger operators defined on a suitable neighbourhood of the vertex $ V_i $.

In this paper we analyze the problem of approximating the free Schr{\"o}dinger evolution on an $ \varepsilon $-thin Dirichlet star-shaped waveguide, with semi-infinite cylindrical ends of thickness $\varepsilon$; in a future publication we shall generalize our results to cover generic waveguides and smooth potentials. In this paper we study convergence of solutions and (weak) convergence of resolvents, and briefly suggest how to add an external potential and study convergence of scattering matrices.

\section{Description of the problem}

In this section, we describe the structure of the waveguide in detail, and we introduce some notation. The waveguide $ \Gamma^\varepsilon $ is obtained by gluing smoothly to the vertex region $ \Omega^\varepsilon $ (a compact set in $\mathbb{R}^3$) a number of semi-infinite cylindrical ends, the branches of the waveguide; the boundary $ \partial \Gamma^\varepsilon $ is smooth\footnote{For illustrative purposes, in the Figures we depict two-dimensional waveguides with edges at some points of the boundary. Note that the hypothesis of smooth boundary could be weakened, and the dimension of the waveguide could be any $ d \geq 2 $.}. The cylindrical branches have two-dimensional section $ \Sigma^\varepsilon $, a compact domain in $ \mathbb{R}^2 $ with smooth boundary, of linear dimension proportional to $\varepsilon$. Also, the linear dimension of the  vertex region is proportional to $\varepsilon$ so that the waveguides $ \lbrace \Gamma^\varepsilon \rbrace $ associated to different values of $\varepsilon$ are connected by a \textit{similarity transformation} (a change in the length scale). With this assumption, the geometry of $\Gamma^\varepsilon$ is essentially fixed, and the problem is equivalent to the computation of the low-energy effective dynamics in a \textit{fixed} domain. Note that this is not true anymore when we have many vertices at a fixed distance, and we shrink the cylinders connecting them.

\subsection*{Structure of the waveguide}

The branches of the waveguide are labeled by the index $j$, which runs from $1$ to $n$. Each branch is isometric to the cylinder $ \Sigma^\varepsilon \times [0, + \infty) $; then
\[ \Gamma^\varepsilon = \Omega ^\varepsilon \cup \left( \bigcup_{j=1}^n \Sigma_j^\varepsilon \times [0, + \infty) \right) \; . \]
The Dirichlet Laplacian on $\Sigma ^\varepsilon$ has discrete spectrum
\[ \left\lbrace \mu_1^\varepsilon, \mu_2^\varepsilon \ldots \right\rbrace \qquad \mu_1^\varepsilon < \mu_2^\varepsilon \leq \ldots \]
and the eigenvalues satisfy the scaling relation
\[ \mu_i^\varepsilon = \frac{\mu_i}{\varepsilon^2} \; . \]

The Dirichlet Laplacian on the entire waveguide has absolutely continuous spectrum which coincides with the semi-infinite interval $[ \mu_1^\varepsilon, + \infty)$, and possibly some discrete eigenvalues $ \lambda^\varepsilon_1 \ldots \lambda^\varepsilon_k $, not larger than $\mu_1^\varepsilon$, corresponding to bound states\footnote{bound states are known to form e.g. near the bent regions of the waveguides; see for example \cite{exner89}} with eigenfunctions $ \phi_1^\varepsilon \ldots \phi_k^\varepsilon $. They also satisfy the scaling rule
\[ \lambda_i^\varepsilon = \frac{\lambda_i}{\varepsilon^2} \; . \]

It is natural to work with the renormalized Hamiltonian
\[ H^\varepsilon = - \Delta_D - \mu_1^\varepsilon \]
acting on $ L^2(\Gamma^\varepsilon) $. In this way, the bound states have negative energy and the absolutely continuous spectrum coincides with the positive half-line.

Let us introduce the crucial concept of \textit{mesoscopic region} $ \Gamma_{int}^{\varepsilon, \ell} $, or sometimes simply $ \Gamma_{int} $. This is obtained by cutting $ \Gamma^\varepsilon $ at distance $ \ell $ from the vertex  region $ \Omega^\varepsilon $:
\[ \Gamma_{int}^{\varepsilon, \ell} \equiv \Omega^\varepsilon \cup \left( \bigcup_{j=1}^n \Sigma_j^\varepsilon \times [0, \ell] \right) \; . \]
The complement of $ \Gamma_{int} $ is referred to as $ \Gamma_{out} $:
\[ \Gamma_{out}^{\varepsilon, \ell} = \bigcup_{j=1}^n \Sigma_j^\varepsilon \times [\ell, + \infty ) \; . \]
It is very convenient to work with the parameter $L$ defined as
\[ L \equiv \ell / \varepsilon \]
instead that with $ \ell $. We refer to $L$ as to the rescaled mesoscopic lenght. With this notation, the mesoscopic region has tickness $ \varepsilon $ and size $ L \varepsilon $.

\begin{remark}\label{remark:MicroscopicAndMacroscopicCoordinates}
To identify the points in the $j$-th branch of the waveguide, we normally use Cartesian coordinates $ (x_j, y_j) $, with values in $ [0, + \infty) \times \Sigma^\varepsilon $. We can say that these coordinates are adapted to the \textit{macroscopic} scale. Later it will be convenient to make use also of rescaled coordinates $ (X_j, Y_j) $, adapted to the \textit{microscopic} scale, ranging in $ [ 0, + \infty) \times \Sigma^1 $, such that
\[ x_j = \varepsilon X_j, \quad y_j = \varepsilon Y_j \; . \]
Notice that in these coordinates, the mesoscopic region is characterized by the inequalities
\[ x_j \leq \ell, \quad X_j \leq L \; . \]
$ \phantom{A} $
\end{remark}

\begin{figure}[htbp]
\centering
\includegraphics[scale=0.7]{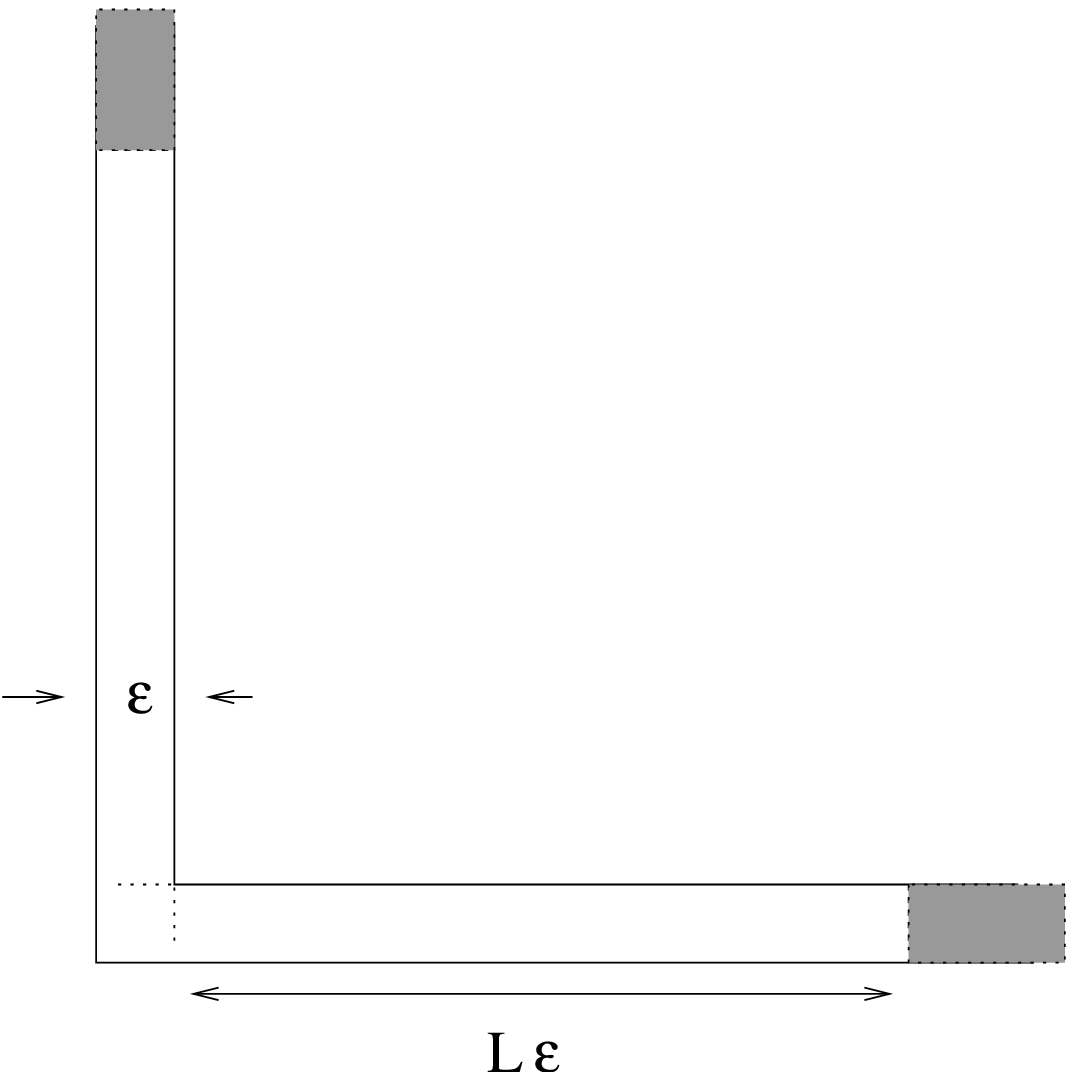}
\caption{A very simple example of waveguide, the L-shaped waveguide. The junction $ \Omega^\varepsilon $ is a square of size $ \varepsilon $, and there are two semi-infinite rectangular branches attached to it. The mesoscopic region $ \Gamma_{int}^{\varepsilon, \ell} $ is white, while $ \Gamma_{out}^{\varepsilon, \ell} $ is shadowed.}\label{fig:L_shaped_waveguide}
\end{figure}

\paragraph*{The general strategy}

The waveguide that we consider is characterized by two parameters: the \textit{microscopic} lenght $ \varepsilon $ and the \textit{mesoscopic} lenght $ L \varepsilon = \ell $. The parameter $ \ell $ is chosen so that the inequality
\[ \varepsilon \ll \ell \ll 1 \]
is satisfied: this motivates our notation, too. In terms of the rescaled mesoscopic lenght $L$, these conditions read
\[ \varepsilon \ll 1 \qquad L \gg 1 \qquad L \cdot \varepsilon \ll 1 \; . \]

Our strategy is the following: given a wavefunction $ \psi $ on $ \Gamma^\varepsilon $, solution of the Schr{\"o}dinger equation, we consider its restriction to $ \Gamma_{int}^{\varepsilon, \ell} $ and $ \Gamma_{out}^{\varepsilon, \ell} $. We prove that $ \psi \upharpoonright \Gamma_{out}^{\varepsilon, \ell} $ essentially factors into $ \Psi \otimes \chi_1^\varepsilon $ plus negligible corrections: $ \Psi $ is the effective wavefunction on the limit graph $ \Gamma $. It is clear that this approach makes sense if $ \ell \ll 1 $, so that $ \Psi $ is defined everywhere except for a small neighbourhood of the vertex.

Our second step is to understand the behaviuor of $ \Psi $ near the vertex. To do this, we study the restriction of $ \psi $ to $ \Gamma_{int}^{\varepsilon, \ell} $. It turns out that \textit{for a suitable class of initial states} (roughly speaking, low-energy states) the behaviour of $ \Psi $ at the vertex does not depend on the initial state, but only on the spectral properties of $ \Gamma_{int}^{\varepsilon, \ell} $. The estimates that constrain the behaviour of $ \Psi $ are expressed in terms of both the parameters $ \ell $, $ \varepsilon $: in concrete situations one can choose $ \ell $ to be a suitable function of $ \varepsilon $, to optimize the error. This is particularly simple for free particles, but we prefer to keep the discussion at a general level to prepare the ground for successive generalizations.

\begin{figure}[htbp]
\centering
\includegraphics[scale=0.7]{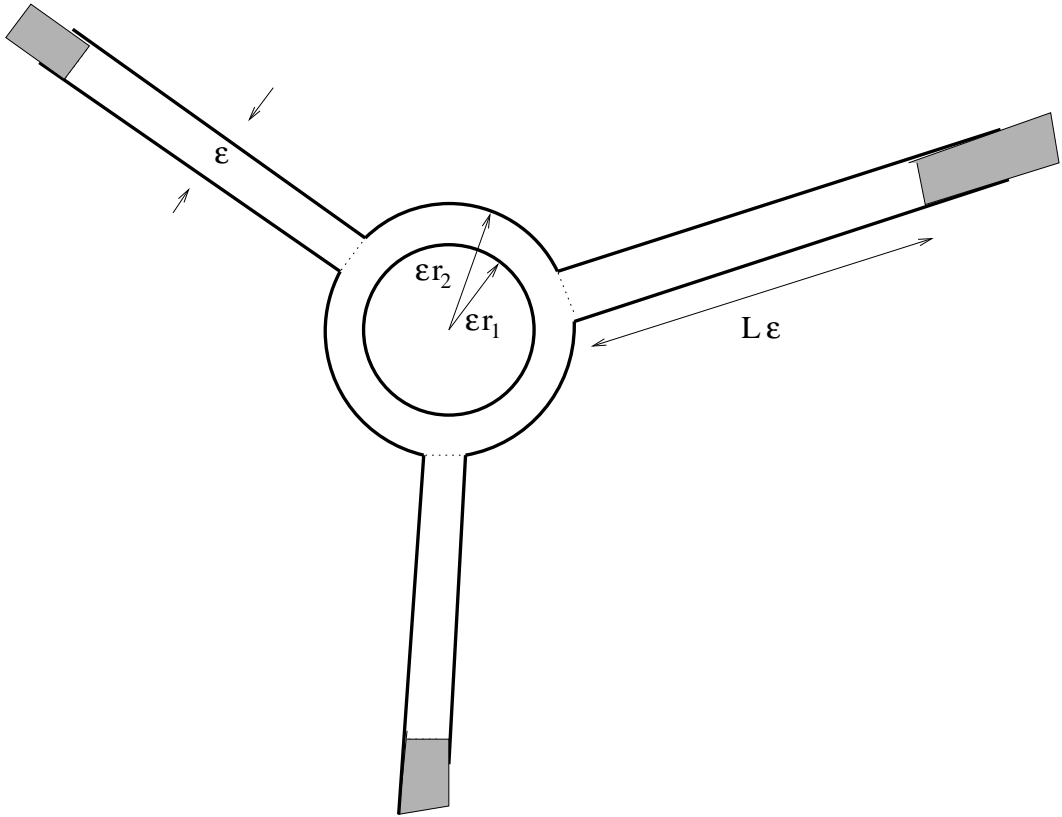}
\caption{A more complicated waveguide. Three branches of thickness $ \varepsilon $ join at a junction $ \Omega^\varepsilon $. In this case, the junction is the inner part of a circular crown, of radii $ r_1 \varepsilon $, $ r_2 \varepsilon $. The mesoscopic region is obtained by attaching to the junction a portion of the leads of lenght $ L \varepsilon $. In general, one may consider even more irregular junctions, as long as they are obtained by rescaling a fixed compact domain by $ \varepsilon $.}\label{fig:junction_with_a_hole}
\end{figure}

\subsection*{Some definitions}

We introduce here the notion of spectrum of the mesoscopic region, and illustrate some results from \cite{grieser08a} that characterize it. Given a generic junction $ \Omega^\varepsilon $ , it is clear that the explicit spectrum cannot be computed in terms of the parameters $ \varepsilon $ and $L$, even if some asymptotics may be calculated for specific geometries. Our point of view is that we can formulate some hypotheses on this spectrum regarding its qualitative behaviour near the continuum treshold, that enable us to distinguish between junctions that lead to decoupling or coupling conditions for the limit Quantum Graph.

\begin{definition}[auxiliary Hamiltonian]
On the mesoscopic region $ \Gamma_{int}^{\varepsilon, \ell} $ we define the operator
\[ H_{int}^{\varepsilon, \ell} = - \Delta - \mu_1^\varepsilon \]
where the Laplace operator $ \Delta $ is defined by Neumann boundary conditions at $ \lbrace \ell \rbrace \times \Sigma_j $ and Dirichlet boundary conditions on the rest of $ \partial \Gamma_{int}^{\varepsilon, \ell} $. $ H_{int}^{\varepsilon, \ell} $ is the \textit{auxiliary Hamiltonian} on the mesoscopic region. The spectrum of $ H_{int}^{\varepsilon, \ell} $ is discrete and it will be denoted by
\[ \sigma \left( H_{int}^{\varepsilon, \ell} \right) = \lbrace \lambda_i^{\varepsilon, \ell} \rbrace_{i \in \mathbb{N}} \; . \]
\definitionend
\end{definition}

The spectrum of the auxiliary Hamiltonian $ H_{int}^{\varepsilon, \ell} $ is sometimes referred to as the spectrum of the mesoscopic region. This spectrum is analyzed in \cite{grieser08a} in great detail, although in a slighly different setting: the domain is rescaled by a factor $ 1 / \varepsilon $ and the eigenvalues are multiplied by $ \varepsilon^2 $, accordingly. The dilated mesoscopic region has linear dimension $L$ and the branches have tickness $1$; sometimes one just says that we are working in coordinates such that $ \varepsilon = 1 $. In particular, it is shown that any eigenvalue (proper and generalized) of $ H^\varepsilon $ can be approximated by an eigenvalue of $ H_{int}^{\varepsilon, \ell} $ and vice-versa. For the point spectrum, the difference between $ \lambda_i^{\varepsilon} $ and $ \lambda_i^{\varepsilon, \ell} $ vanishes exponentially in $ L = \ell/\varepsilon $: in our notation, it is $ \mathcal{O} ( \varepsilon^{-2} e^{-L} ) $. The eigenfunctions are shown to converge too, in the proper topology, with the same speed.

An immediate consequence of the considerations above is that the first $k$ eigenvalues of $ H_{int}^{\varepsilon, \ell} $ will converge to the energy levels $ \lambda_1^\varepsilon \ldots \lambda_k^\varepsilon $ corresponding to the bound states of $ \Gamma^\varepsilon $; the other eigenvalues of $ H_{int}^{\varepsilon, \ell} $ will ``merge into the continuum''.

For each value of $ \varepsilon $, we shall divide the eigenvalues that merge into the continuous spectrum in two classes, according to qualitative properties of the approach to the bottom of the continuous spectrum when $ \varepsilon \to 0 $ and $ L \to \infty $. If there are no elements in the first class, the limit motion on the graph has Dirichlet (decoupling) boundary conditions at the vertex of the graph. The eigenfunctions corresponding to the elements of the first class (we call them \textit{resonant sequences}) provide ``bridges'' connecting the different branches of the graph; if there is only one element in this class, the motion on the limit graph is generated by a Laplacian with suitable Kirchhoff-like boundary conditions, which depend on the asymptotic form of this connecting eigenfunction. If there are more than one elements, the problem becomes more complex, and will not be treated here.

Recall that the first $k$ eigenvalues $ \lambda_1^{\varepsilon, \ell} \ldots \lambda_k^{\varepsilon, \ell} $ of $ H_{int}^{\varepsilon, \ell} $ converge to the $ k $ bound states $ \lambda_1^\varepsilon \ldots \lambda_k^\varepsilon $ of the waveguide as $ \varepsilon \to 0 $; henceforth, we will loosely refer to all these states as ``bound states'', or ``localized states''. To proceed further, it will be convenient to introduce the following

\begin{definition}[spectral gap condition]\label{def:SpectralGapCondition}
The family of (mesoscopic) regions $ \Gamma_{int}^{\varepsilon, \ell} $ satisfies the spectral gap condition if there exist $ \mu_0 > 0 $ and $ \gamma > 0 $ such that the eigenvalues of the non-bound states $ \psi_{n}^{\varepsilon, \ell} $, $ n \geq k+1 $ satisfy, uniformly in $ \varepsilon > 0 $
\begin{equation}\label{eq:SpectralGapCondition}
\lambda_{n}^{\varepsilon, \ell} \geq \frac{ \mu_0 }{ \varepsilon^2 } \cdot \left( \frac{1}{L} \right)^\gamma, \qquad n \geq k+1 \qquad L \equiv \frac{\ell}{\varepsilon} \; .
\end{equation}
\definitionend
\end{definition}

We are also interested in waveguides where the gap condition is not satisfied, i.e. there are some non-bound states with eigenvalues that do not satisfy the inequality stated above. We do not treat the general case, but limit ourselves to the following important special case:

\begin{definition}[resonant sequence condition]\label{def:ResonantSequenceCondition}
The family of (mesoscopic) regions $ \Gamma_{int}^{\varepsilon, \ell} $ satisfies the resonant sequence condition if the eigenvalues of the non-bound states $ \psi_{n}^{\varepsilon, \ell} $, $ n \geq k+2 $ (hence, all of them but the $ (k+1) $-th) satisfy condition (\ref{eq:SpectralGapCondition}) uniformly in $ \varepsilon > 0 $, while
\begin{equation}\label{eq:ResonantStateCondition}
- \frac{ \mu_0' }{ \varepsilon^2 } \cdot \left( \frac{1}{L} \right)^{\gamma'} \leq \lambda_{k+1}^{\varepsilon, \ell} \leq 0
\end{equation}
for some $ \gamma' \geq 1 $, $ \mu_0' > 0 $.\\
\definitionend
\end{definition}

We call the sequence of states $ \lbrace \psi_{k+1}^{\varepsilon, \ell} \rbrace $ which characterize the definition above a \textit{resonant sequence} of states. The name we have chosen comes from the fact that if the waveguide has a zero-energy resonant sequence (a solution of the Dirichlet problem not in $ L^2 (\Gamma^\varepsilon) $), then its restriction to $ \Gamma_{int}^{\varepsilon, \ell} $ provides a resonant sequence. The last statement follows from a variational argument, which bounds the eigenvalue from above, and a result discussed in \cite{grieser08a} which bounds the eigenvalue from below. Note that \cite{grieser08a} proves that in the self-similar case the convergence is exponentially fast. Indeed, including the polynomial bound is useful for generalization to the non-uniform case.

In a self-similar family of waveguides, if a resonance exists for one value of $ \varepsilon $, it exists for all values of $ \varepsilon $. But a resonant sequence may exist even if there is no zero-energy resonance: this will be very useful for generalising the analysis to graphs with many vertices, where self-similarity is lost. We expect the resonant state condition to hold when the waveguide has some resemblance with the domain accessible to conduction electrons of an aromatic molecule\footnote{The role of a resonant structure has been advocated by B. Pavlov, from a different point of view.}.

\begin{remark}
The inequalities of Equations \ref{eq:SpectralGapCondition}, \ref{eq:ResonantStateCondition} are written in terms of the parameters $ \varepsilon $, $L$ so that their application in the forthcoming Theorems is immediate. But we stress that if we rewrite the Definitions above \textit{for the rescaled mesoscopic region} of tickness $1$ and lenght $L$ (``take $ \varepsilon = 1 $''), for which the spectrum is rescaled by a factor $ \varepsilon^2 $, then the eigenvalues are functions of the parameter $L$ alone, and Equations \ref{eq:SpectralGapCondition}, \ref{eq:ResonantStateCondition} become conditions on the asymptotic behaviour of these eigenvalues, for $L$ large. Thus, the information about the limit Quantum Graph (wether the junction gives coupling or decoupling gluing conditions) is encoded into the asymptotics of the spectrum of the rescaled mesoscopic region. The parameter $ \varepsilon $ will play a role when we discuss the class of states (the low-energy states) that can be approximated by effective wavefunctions on the graph; at that stage, it will be convenient to choose $L$ as a function of $ \varepsilon $ in order to minimize the error terms.
\end{remark}

\paragraph*{Towards the construction of examples}

It would be very interesting to build some concrete examples of geometries such that the Definitions above are fulfilled. We give some suggestions that come from preliminary computations.

We expect that the L-shaped waveguide of Figure \ref{fig:L_shaped_waveguide} falls into Definition \ref{def:SpectralGapCondition}. This waveguide has one bound state \cite{exner89} (in our notation, $ k = 1 $), hence we need to look at the behaviour of the second eigenvalue $ \lambda_2^{\varepsilon, \ell} $ of the auxiliary Hamiltonian. For this particular geometry it should be possible to compute the asymptotic of the spectrum for large $L$, by adaptating the arguments used in the paper \cite{grieser07}: we expect that in this case the condition of Equation \ref{eq:SpectralGapCondition} is fulfilled. Note that if $ \psi_2^{\varepsilon, \ell} $ is antisymmetric with respect to the reflection through the axis of the waveguide, this follows immediately from a standard bracketing argument.

We also expect that the waveguide depicted in Figure \ref{fig:junction_with_a_hole} falls into Definition \ref{def:ResonantSequenceCondition} instead, at least for a wise choice of the radii $ r_1, r_2 $ of the junction. The idea is that there are no bound states, so that if we are able to tune the parameters in such a way that a resonance for the infinite waveguide is present, then the resonance induces a resonant sequence by the mechanism described previously. Note that the ``hole'' is a suggestion that comes from Physics: it mimics a repulsive potential, which is given for example by the core electrons of the carbon atom that sits at the vertex of a graphene layer.

Clearly, it is more difficult to exihibit examples of guides satisfying this hypothesis. It is to be expected that a resonance (or its counterpart, the resonant sequence) is not present for a generic waveguide, accordingly to the general wisdom that decoupling generically occurs for thin Dirichlet waveguides \cite{molchanov07} \cite{molchanov08} \cite{grieser08a}.

\section{Multiscale decoupling}

In this section, we want to show that \textit{outside the mesoscopic region} low-energy wavefunctions are confined to the first transverse mode, and the coefficient solves the free one-dimensional Schr{\"o}dinger equation; therefore, they can be identified with one-dimensional wavefunctions on $ \Gamma $. The next sections are devoted to describe the behaviour of the solutions inside the mesoscopic region, that is to say, near the vertex of $ \Gamma $.

Let $ \psi^{\varepsilon} $ belong to the image of the spectral projector $ P_{(0, E]} $ associated to $ H^\varepsilon $. For practical purposes, we may assume that at time zero $ \psi^\varepsilon $ is approximated by a product state
\[ \psi^\varepsilon (t = 0) \simeq \Phi \otimes \chi_1^\varepsilon \]
supported in $ \Gamma_{out}^{\varepsilon, \ell} $. This state has finite energy, and is orthogonal to the bound states $ \phi_1^\varepsilon \ldots \phi_k^\varepsilon $ of $ \Gamma^\varepsilon $: both these properties hold at all times, by unitarity of the time evolution. At any time, the restriction of $ \psi^\varepsilon $ to $ \Gamma_{out}^{\varepsilon, \ell} $ admits a convergent expansion in Fourier modes:
\[ \psi^{\varepsilon} \upharpoonright_{ \Gamma_{out}^{\varepsilon, \ell} } = \sum_{m=1}^\infty \Psi^\varepsilon_m \otimes \chi_m^\varepsilon \]
where $ \chi^\varepsilon_m $ is the $m$-th Dirichlet eigenfunction of $-\Delta$ on $\Sigma^\varepsilon$, and $ \Psi^\varepsilon_m $ is a vector-valued function on $ \mathbb{R}^+ $; its components are $ \Psi^{\varepsilon}_{m, j} $, the index $ j $ runs from $ 1 $ to $ n $ and labels the branches of the waveguide. It is useful to think about $ \Psi^\varepsilon_m $, for any $ m $, as a function on the limiting graph $ \Gamma $.

We will prove the following

\begin{theorem}[Multiscale decoupling] \label{thm:MultiscaleDecoupling}
Consider the wavefunction $ \psi^\varepsilon $ described above, and restrict it to $ \Gamma_{out}^{\varepsilon, \ell} $. We can prove that
\[ \sum_{m=2}^\infty \Vert \Psi_m^\varepsilon \Vert^2_{ L^2 (\Gamma) } \leq C \cdot \varepsilon^2 \]
for some numerical constant $ C $. Moreover,
\[ \sum_{m=2}^\infty \Vert \Psi_m^\varepsilon \Vert^2_{ W^{1, 2} (\Gamma) } \leq C' \cdot \varepsilon, \qquad \sum_{m=2}^\infty \Vert \Psi_m^\varepsilon \Vert^2_{ L^\infty (\Gamma) } \leq C' \cdot \varepsilon \]
for some numerical constant $ C' $.\\
\theoremend
\end{theorem}

In other words, the components of the wavefunction $\psi^{\varepsilon}$ along the higher transverse modes are suppressed in the limit $\varepsilon \rightarrow 0$: the only component which survives is ``frozen'' in the first transverse mode $\chi_1^\varepsilon$.

\begin{proof}
The first $k$ eigenvalues $ \lambda_1^{\varepsilon, \ell} \ldots \lambda_k^{\varepsilon, \ell} $ of the auxiliary Hamiltonian $ H_{int}^{\varepsilon, \ell} $ converge to the bound states $ \lambda_1^\varepsilon \ldots \lambda_k^\varepsilon $ in the sense explained above: the difference is $ \mathcal{O} ( \varepsilon^{-2} e^{-L} ) $ (we will loosely refer to them as ``bound states'' too); the eigenvectors $ \psi_1^{\varepsilon, \ell} \ldots \psi_k^{\varepsilon, \ell} $ converge uniformly, with all derivatives, to the respective bound states $ \phi_1^\varepsilon \ldots \phi_k^\varepsilon $ of $ \Gamma^\varepsilon $ restricted to $ \Gamma_{int}^{\varepsilon, \ell} $. It follows that if $ \psi $ is orthogonal to $ \phi_1^\varepsilon \ldots \phi_k^\varepsilon $, its restriction to $ \Gamma_{int}^{\varepsilon, \ell} $ will be almost orthogonal to $ \psi_1^{\varepsilon, \ell} \ldots \psi_k^{\varepsilon, \ell} $, up to an error that is $ \mathcal{O} ( \varepsilon^{-2} e^{-L} ) $.

In the following we will just say that $ \psi $ (restricted to $ \Gamma_{int}^{\varepsilon, \ell} $) is almost orthogonal, or quasi-orthogonal, to $ \psi_1^{\varepsilon, \ell} \ldots \psi_k^{\varepsilon, \ell} $, without writing down the exponentially small error: this is because we will obtain estimates that are polynomial in $ \varepsilon $ and $ \ell $, so that the exponentially small error is easily reabsorbed; and this will help a lot in keeping the formulas clean.

If the Laplacian on $ \Gamma_{int}^{\varepsilon, \ell} $ has no eigenvalues below $ \mu_1^\varepsilon $ other than the ``bound states'', a standard spectral argument implies that
\[ \int_{ \Gamma_{int}^{\varepsilon, \ell} } \left( | \nabla \psi^\varepsilon |^2 - \mu_1^\varepsilon | \psi^\varepsilon |^2 \right) \geq 0 \; . \]
Moreover, by substituting the Fourier expansion of $ \psi^\varepsilon \upharpoonright \Gamma_{out}^{\varepsilon, \ell} $ we easily see that
\[ \int_{ \Gamma_{out}^{\varepsilon, \ell} } \left( | \nabla \psi^\varepsilon |^2 - \mu_1^\varepsilon | \psi^\varepsilon |^2 \right) \geq 0 \; . \]
Hence, both these contributions are positive: from the energy bound on $ \psi^\varepsilon $ we know that their sum $ \int_{ \Gamma^\varepsilon } \left( | \nabla \psi^\varepsilon |^2 - \mu_1^\varepsilon | \psi^\varepsilon |^2 \right) $ is positive and smaller than $E$, and so each one of them (in particular, the integral over $ \Gamma_{out} $) is bounded by $E$. Notice that if the integral over $ \Gamma_{int} $ had been large and negative, the integral over $ \Gamma_{out} $ could have been arbitrarily large; this is why we must avoid bound states.

After this simple, but important observation, the first part of the theorem easily follows by substituting the Fourier decomposition of $ \psi^{\varepsilon} \upharpoonright { \Gamma_{out}^{\varepsilon, \ell} } $ in the integral. We obtain
\[ \int_{ \Gamma_{out}^{\varepsilon, \ell} } | \psi^{\varepsilon} - \left( \Psi_1^\varepsilon \otimes \chi_1^\varepsilon \right) |^2 \leq \left( \frac{E}{\mu_2 - \mu_1} \int_{ \Gamma^\varepsilon } | \psi^{\varepsilon} |^2 \right) \cdot \varepsilon^2 \; . \]
To complete the theorem, note that $ \Vert H^\varepsilon \psi^\varepsilon \Vert^2 \leq E^2 $. Repeating the previous argument, with simple modifications, one proves that the sum of the norm squared of $ \Delta \Psi^\varepsilon_m $ is bounded: by interpolation, the sum of the squared norms of $ \nabla \Psi^\varepsilon_m $, for $ m \geq 2 $ is $ \mathcal{O} (\varepsilon) $. This implies that the $ W^{1, 2} $-norm of $ \Psi^\varepsilon_m $, $ m \geq 2 $, is $ \mathcal{O} (\varepsilon^{1/2}) $: we recall that this implies that the $ L^\infty $-norm is infinitesimal too, with the same bounds.
\end{proof}

\begin{remark}
Consider the expectation value of an observable, which is either given by the multiplication by a smooth function supported in $ \Gamma_{out}^{\varepsilon, \ell} $, or a linear differential operator in the longitudinal derivatives $ \partial / \partial x_j $, $ j = 1 \ldots n $ with coefficients in the smooth functions supported in $ \Gamma_{out}^{\varepsilon, \ell} $. Theorem \ref{thm:MultiscaleDecoupling} tells us that the expectation values on $ \psi^\varepsilon $ of these observables converge for small $ \varepsilon $ to the expectation values on $ \Psi_1^\varepsilon $ of the corresponding observables restricted to $ \Gamma $. Consider for instance a function $ \phi $ on $ \Gamma^\varepsilon $ (vanishing on $ \Gamma_{int}^{\varepsilon, \ell} $): and say that we want to compute
\[ \int_{ \Gamma^\varepsilon } \phi \cdot | \psi^\varepsilon |^2, \quad \int_{\Gamma^\varepsilon} \phi \cdot ( \psi^\varepsilon )^* \partial_x \psi^\varepsilon. \]
By Theorem \ref{thm:MultiscaleDecoupling}, it is easy to see that apart from an error term vanishing as a suitable power of $ \varepsilon $, these expectation values respectively converge to
\[ \int_{ \Gamma } \phi_\Gamma \cdot | \Psi_1^\varepsilon |^2, \quad \int_{ \Gamma } \phi_\Gamma \cdot ( \Psi_1^\varepsilon )^* \partial_x \Psi_1^\varepsilon, \]
where $ \phi_\Gamma $ is the evaluation of $ \phi $ on $ \Gamma $. Notice that $ \Vert \chi_1^\varepsilon \Vert^2 = 1 \; \forall \varepsilon $.
\end{remark}

The last remark is crucial: it allows us to use $ \Psi_1^\varepsilon $ instead of $ \psi^\varepsilon $ to compute the expectation values of important observables, such as the density and (longitudinal) current of particles. Henceforth, we will study the fundamental Fourier component $ \Psi_1^\varepsilon $, and systematically neglect the higher-order contributions to $ \psi^\varepsilon $, which do not give contribution to the observables of interest.

It is desirable to determine an effective evolution equation for $ \Psi_1^\varepsilon $ in the limit $ \varepsilon \to 0 $. Notice that $ \Psi_{1, j}^\varepsilon $, defined as
\[ \Psi_{1, j}^\varepsilon (x) = \int_{ \Sigma_j^\varepsilon } \psi^\varepsilon (x, y) \cdot \chi_1^\varepsilon (y) \, \mathsf{d}y \]
is measurable, as its time and space derivatives; the $ L^2 $-norms of $\Psi_1^\varepsilon$ and of $\nabla \Psi_1^\varepsilon$ are bounded, hence $\Psi_1^\varepsilon$ is in $ W^{1, 2} (\Gamma) $. It is easy to see that $ \Psi_{1, j}^\varepsilon $ satisfies the free (one-dimensional) Schr{\"o}dinger equation, in weak sense. Indeed, for any smooth test function $\varphi$ with support in $ \Gamma_{out}^{\varepsilon, \ell} $, and of the form $ \Phi \otimes \chi_1^\varepsilon $ the equations of motion read
\[ i \frac{\mathsf{d}}{\mathsf{d}t} \langle \varphi, \psi^\varepsilon \rangle = \langle \nabla \varphi, \nabla \psi^\varepsilon \rangle - \mu_1^\varepsilon \langle \varphi, \psi^\varepsilon \rangle \]
and after the integration on the transverse variable the equation reduces to
\[ i \frac{\mathsf{d}}{\mathsf{d}t} \langle \Phi, \Psi_1^\varepsilon \rangle = \langle \nabla \Phi, \nabla \Psi_1^\varepsilon \rangle \; . \]
This equation is the weak form of the free evolution of $ \Psi_1^\varepsilon $ along the edges of the graph. Since the support of $ \Phi $ is at a distance not smaller than $ \ell $ from the vertex, the equation is satisfied \textit{outside} the vertex. By a compactness argument, there exists a subsequence $ \lbrace \Psi_1^{\varepsilon_n} \rbrace_{ n \in \mathbb{N} } $ having a weak limit $ \Psi \in W^{1, 2}(\Gamma) $, which must satisfy the weak Schr{\"o}dinger equation on the edges of $ \Gamma $ as well (recall that the initial datum is the same for all values of $ \varepsilon $!). $ \Psi $ is our limit wavefunction on $ \Gamma $: once we have proven that the unitary evolution of $ \Psi $ is uniquely determined, it will be \textit{a posteriori} clear that any subsequence, and hence the whole sequence $ \lbrace \Psi_1^\varepsilon \rbrace $ converges to the same $ \Psi $. Our goal is to use $ \Psi $, instead of $ \Psi_1^\varepsilon $, to approximate the expectation values of $ \psi^\varepsilon $. The next section is devoted to determining the time evolution of $ \Psi $, which depends crucially upon the behaviour of $ \Psi $ at the vertex of $ \Gamma $.

\section{The limit motion on the graph}

In the following, we will show how one can  approximate the low-energy dynamics of the dynamical system $ (\Gamma^\varepsilon, H^\varepsilon) $ with a suitable dynamics $ H_\Gamma $ on the metric graph $ \Gamma $. The cornerstone of our approach is the study of the spectrum in the mesoscopic region $ \Gamma_{int}^{\varepsilon, \ell} $. The limit dynamics is uniquely identified in this section; more quantitative results (convergence of resolvents and of wave operators) are dealt with in the next sections.

For simplicity, we only treat the case with no external potential. The method can be used also in case there is a continuous potential $ V(x), \; x \in \mathbb{R}^3 $; we shall denote by $ V_\Gamma $ its restriction to the graph. We suggest that such perturbations can be treated by approximating the wave operators $ W^\pm (H^\varepsilon + V, H^\varepsilon) $ on the waveguide $ \Gamma^\varepsilon $ with their analogues $ W^\pm (H_\Gamma + V_\Gamma, H_\Gamma) $ on $ \Gamma $: this is sketched in the last section.

\subsection*{The method of approximating (zero-energy) resonant sequences}

We have a sequence of star-shaped waveguides $ \lbrace \Gamma^\varepsilon \rbrace $ (``$\varepsilon$-fat'' star graphs) which shrink to a star-shaped graph $ \Gamma $ as $ \varepsilon \to 0 $. The method we describe allows the determination of the boundary conditions at the vertex of $ \Gamma $, \textit{which are needed in order to define the effective dynamics on the graph} that, for small $ \varepsilon $, approximates the low-energy dynamics generated by $ H^\varepsilon $. These boundary conditions depend on the shape of $ \Gamma^\varepsilon $ in the neighborhood of the vertex region $ \Omega^\varepsilon $. The dependence is through the spectra of the auxiliary operators $ H_{int}^{\varepsilon, \ell} $ acting on the mesoscopic region $ \Gamma_{int}^{\varepsilon, \ell} $.

We will see in the next paragraph that the control of the effective wavefunction $ \Psi $ on the limit graph $ \Gamma $ near the vertex is possible if we consider a particular class of initial states (low-energy states) and choose wisely the parameters $ \varepsilon $, $L$. But once we restrict to this class of initial states, the gluing conditions at the vertex of the graph depend on the geometry of the waveguide and nothing else: the dependence is through the asymptotics (in the parameter $L$) of the eigenvalues of the rescaled mesoscopic region.

\subsection*{The wavefunctions near the vertex}

Let $ \psi^\varepsilon \in L^2 (\Gamma^\varepsilon) $, smooth, normalized, energy-bounded state orthogonal to the bound states of $\Gamma^\varepsilon$. Consider the finite cylinders
\[ B_j^{\varepsilon, \ell} = \Sigma_j^\varepsilon \times [0, \ell] \quad j=1 \ldots n \]
contained in the mesoscopic region: $ \lbrace \ell \rbrace \times \Sigma_j^\varepsilon $ is part of the boundary of the mesoscopic region. Notice that (up to a zero-measure boundary)
\[ (\cup_j B_j^{\varepsilon, \ell} ) \cap \Omega^{\varepsilon} = \emptyset \; . \]

The function $ \psi^\varepsilon \upharpoonright { B_j^{\varepsilon, \ell} } $ can be Fourier decomposed as usual:
\[ \psi^\varepsilon \upharpoonright_{ B_j^{\varepsilon, \ell} } = \left( \Psi_{1, j}^\varepsilon \otimes \chi_1^\varepsilon \right) \upharpoonright_{B_j} + R^{\varepsilon}_{1,j}  \; . \]
We can prove the following

\begin{theorem} \label{thm:GettingDirichletConditions}
Assume that the spectral gap condition holds. If $\psi^\varepsilon $ belongs to $ P_{ (0, E] } L^2 (\Gamma^\varepsilon) $, then the $L^\infty$-norm of
\[ \Psi^\varepsilon = \Psi_{1, j}^\varepsilon \upharpoonright_{[0, \ell]} \]
is bounded by $ C_1 \cdot \ell^{1/2} + C_2 \cdot \ell^{1/2} L^{-(2-\gamma)/2} $, where $ C_1 $, $ C_2 $ are numerical constants.\\
\theoremend
\end{theorem}

An important consequence of this theorem is the following. Let $ \Psi $ be the weak limit of any convergent subsequence $ \lbrace \Psi_1^{\varepsilon_n} \rbrace $. It is in $ W^{1, 2} (\Gamma \backslash \lbrace V \rbrace) $, therefore each component $ \Psi_j $, $ j = 1 \ldots n $ is right-continuous at the vertex $V$. Since its approximants are uniformly bounded near the vertex by an infinitesimal quantity, it follows that
\[ \lim_{ x \to 0^+ } \Psi_j (x) = 0 \qquad j = 1 \ldots n \; . \]

\begin{proof}
$\psi^\varepsilon$ is a state with energy bounded by $E$. Application of the spectral condition, together with (quasi)-orthogonality to the negative-energy states of $ H_{int}^{\varepsilon, \ell} $\footnote{we recall that these states approximate the bound states of $ \Gamma^\varepsilon $ up to an exponentially small error, see the Definition of $ H_{int}^{\varepsilon, \ell} $. We omit the exponentially small errors from the formulas, to make them easier to read.}, tells us that
\[ \int_{\Gamma_{int}} | \nabla \psi^\varepsilon |^2 \geq \left[ \mu_1^\varepsilon + \frac{\mu_0}{\varepsilon^2} \left( \frac{\varepsilon}{\ell} \right)^\gamma \right] \int_{\Gamma_{int}} | \psi^\varepsilon |^2, \qquad \mu_0 > 0, \; \gamma > 0 \; . \]
This inequality can be written in a more useful form:
\[ \frac{\mu_0}{\varepsilon^2} \left( \frac{\varepsilon}{\ell} \right)^\gamma \int_{\Gamma_{int}} | \psi^\varepsilon |^2 \leq \int_{\Gamma_{int}} \left( | \nabla \psi^\varepsilon |^2 - \mu_1^\varepsilon | \psi^\varepsilon |^2 \right) \; . \]
Now, we recall that
\[ \int_{\Gamma_{out}} \left( | \nabla \psi^\varepsilon |^2 - \mu_1^\varepsilon | \psi^\varepsilon |^2 \right) \geq 0 \]
(this is easily seen by expanding $\psi^\varepsilon$ in Fourier series) and so
\[ \int_{\Gamma_{int}} \left( | \nabla \psi^\varepsilon |^2 - \mu_1^\varepsilon | \psi^\varepsilon |^2 \right) \leq \int_{\Gamma^\varepsilon} \left( | \nabla \psi^\varepsilon |^2 - \mu_1^\varepsilon | \psi^\varepsilon |^2 \right) \leq E \]
from which we deduce
\[ \int_{\Gamma_{int}^{\varepsilon, \ell}} | \psi^\varepsilon |^2 \leq \frac{E}{\mu_\gamma} \varepsilon^2 \left( \frac{\ell}{\varepsilon} \right)^\gamma \; . \]
This states that the $L^2$ norm of $\psi \upharpoonright_{\Gamma_{int}}$ (and hence of $\psi \upharpoonright_{B_j}$) is bounded by a suitable power of $\varepsilon$ (not uniformly in  $E$).

For any smooth function $f$ on the interval $[a, b]$ and any point $x_0 \in [a, b]$, the following Poincar\'e-like estimate holds:
\[ (b-a) \left[ | f(x_0) | - \Vert \nabla f \Vert_{L^2} \sqrt{b-a} \right]^2 \leq \int_a^b | f(x) |^2 \mathsf{d}x \; . \]

Consider the restriction of $\psi^\varepsilon$ to $B_j^{\varepsilon, \ell}$: we write
\[ \psi^\varepsilon \upharpoonright_{ B_j^{\varepsilon, \ell} } = \left( \Psi_{1, j}^\varepsilon \otimes \chi_1^\varepsilon \right) \upharpoonright_{ B_j^{\varepsilon, \ell} } + R_{1, j}^\varepsilon \equiv \Psi^\varepsilon \otimes \chi_1^\varepsilon + R_{1, j}^\varepsilon \; . \]

The remainder $ R_{1, j}^\varepsilon $ is the sum of the components of $\psi^\varepsilon$ along the higher transverse modes. From the Poincar\'e estimate applied to $\Psi^\varepsilon$,
\[ \ell \left[ | \Psi^\varepsilon (x_0) | - \Vert \nabla \Psi^\varepsilon \Vert_{L^2} \cdot \ell^{1/2} \right]^2 \leq \int_0^\ell | \Psi^\varepsilon |^2 \leq \int_{B_j} | \psi^\varepsilon |^2 \; ; \]
Since
\[ \Vert \nabla \Psi^\varepsilon \Vert^2_{L^2} \leq \int_{ B_j^{\varepsilon, \ell} } \left( | \nabla \psi^\varepsilon |^2 - \mu_1^\varepsilon | \psi^\varepsilon |^2 \right) \leq E \]
is finite, we obtain an estimate  for $ \sup | \Psi^\varepsilon | \equiv \Vert \Psi^\varepsilon \Vert_\infty $ in terms of $\ell$, $\varepsilon$:
\[ | \Psi^\varepsilon (x_0) | \leq E \cdot \ell^{1/2} + \sqrt{E/\mu_0} \cdot \varepsilon^{1/2} \left( \frac{1}{L} \right)^{\frac{1 - \gamma}{2}} \; . \]

Using $ \varepsilon^{1/2} = \ell^{1/2} \cdot L^{-1/2} $ this last equation can be rewritten as:
\[ | \Psi^\varepsilon (x_0) | \leq \ell^{1/2} \left\lbrace E + \sqrt{E/\mu_0} \cdot \left( \frac{1}{L} \right)^{ \frac{2 - \gamma}{2} } \right\rbrace \; . \]
If $ \gamma \leq 2 $, the term in curly brackets is bounded by $ E + \sqrt{E/\mu_\gamma} $ as $ \varepsilon \to 0 $, irrespective of the particular choice of the parameter $L$ that we have introduced. If $ \gamma > 2 $ then we need to restrict the choice of $L$, in order to bound $ | \Psi^\varepsilon (x_0) | $ with an infinitesimal quantity. In fact, by choosing $ L = \varepsilon^{- \alpha} $ for $ \alpha $ sufficiently close to $0$, the last equation becomes
\[ | \Psi^\varepsilon (x_0) | \leq E \cdot \varepsilon^{1 - \alpha} + \sqrt{E/\mu_0} \cdot \varepsilon^{1 - \alpha \gamma / 2} \]
and the right hand side clearly vanishes in the limit $ \varepsilon \to 0 $ for any value of $ \gamma $, if $ \alpha $ is suitably chosen.
\end{proof}

In the following it will be useful to know something about the restriction of $ \psi^\varepsilon $ to the boundary of $ \Gamma_{int}^{\varepsilon, \ell} $:

\begin{theorem}\label{thm:HigherModesOfPsiAtTheBoundary}
Assume that the spectral gap condition holds. If $\psi^\varepsilon $ belongs to $ P_{ (0, E] } L^2 (\Gamma^\varepsilon) $, then the $ L^2 $-norm of
\[ \left( \psi^\varepsilon - \Psi_1^\varepsilon \otimes \chi_1^\varepsilon \right) \upharpoonright_{ \partial \Gamma_{out}^{\varepsilon, \ell} } \]
is bounded by $ C_1 \cdot \ell^{1/2} + C_2 \cdot \ell^{1/2} L^{-(2-\gamma)/2} $, where $ C_1 $, $ C_2 $ are numerical constants. The constants are not necessarily the same of theorem \ref{thm:GettingDirichletConditions}.\\
\theoremend
\end{theorem}

\begin{proof}
Consider \textit{the sum of} the higher modes $ \Psi_2^\varepsilon, \, \Psi_3^\varepsilon \dots $: using $\Psi_m^\varepsilon$ in the Poincar\'e-like estimate and summing up from $m=2$ to $\infty$, we get
\[ \ell \sum_{m=2}^\infty  \left( | \Psi_m^\varepsilon (x_0) | - \ell^{1/2} \int \Vert \nabla \Psi_m^\varepsilon \Vert \right)^2 \leq \sum_{m=2}^\infty \Vert \Psi_m^\varepsilon \Vert^2 \; . \]
(all norms are $L^2$-norms). We use the bounds obtained in Theorem \ref{thm:MultiscaleDecoupling} for the sum of the higher modes: the proof then mimics the proof of theorem \ref{thm:GettingDirichletConditions}.
\end{proof}

Another simple modification of Theorem \ref{thm:GettingDirichletConditions} allows to state that, if the resonant state condition holds:

\begin{theorem}\label{thm:RelaxOnResonantSequence}
Assume that the resonant sequence condition holds. If $\psi^\varepsilon $ belongs to $ P_{ (0, E] } L^2 (\Gamma^\varepsilon) $, define
\[ c_\Psi^\varepsilon \equiv \frac{ \langle \psi_{k+1}^{\varepsilon, \ell}, \psi^\varepsilon \rangle }{ \langle \psi_{k+1}^{\varepsilon, \ell}, \psi_{k+1}^{\varepsilon, \ell} \rangle } \; , \]
and
\[ \left( \psi^\varepsilon - c_\Psi^\varepsilon \cdot \psi_{k+1}^{\varepsilon, \ell} \right) \upharpoonright_{ B_j^{\varepsilon, \ell} } \equiv \tilde{\Psi}^\varepsilon \otimes \chi_1^\varepsilon + \tilde{R}_j^\varepsilon \; . \]
Then, the $ L^\infty $ norm of $ \tilde{\Psi}^\varepsilon $, restricted to the segment $ [0, \ell] $, is bounded by $ C_1 \cdot \ell^{1/2} + C_2 \cdot \ell^{1/2} L^{-(2-\gamma)/2} $, where $ C_1 $, $ C_2 $ are numerical constants. Recall that $ \lbrace \psi_{k+1}^{\varepsilon, \ell} \rbrace $ is the resonant sequence of $ \Gamma_{int}^{\varepsilon, \ell} $.\\
\theoremend
\end{theorem}

Pictorially, we may say that $\psi^\varepsilon$ ``relaxes'' to $ \psi_{k+1}^{\varepsilon, \ell} $.

\subsection*{Dirichlet dynamics on graphs }

In the spectral gap case, our argument is already sufficient to establish that the limit dynamics on the star graph is the decoupled one with Dirichlet boundary conditions at the vertex. Let $ \psi^\varepsilon (t) $ be a solution of the Schr{\"o}dinger equation in $ P_{ (0, E] } L^2 (\Gamma^\varepsilon) $. We know that
\[ \psi^\varepsilon  (t) \upharpoonright_{ \Gamma_{out}^{\varepsilon, \ell} } = \Psi_1^\varepsilon (t) \otimes \chi_1^\varepsilon + R_1^\varepsilon (t) \; . \]
To our purposes, $ R_1^\varepsilon $ can be neglected. The function $ \Psi_1^\varepsilon $ has a weak limit $ \Psi $, for $ \varepsilon \to 0 $; the limit exists by compactness (choose a subsequence), and it satisfies the weak Schr{\"o}dinger equation on $ \Gamma \backslash \lbrace V \rbrace $. By Theorem \ref{thm:GettingDirichletConditions}, $ \Psi $ satisfies Dirichlet boundary conditions at the vertex $V$:
\[ \Psi_j (0) = 0 \qquad j = 1 \ldots n \]
and this fixes the unitary evolution of $ \Psi $ on the limit graph $ \Gamma $ completely.

Since \textit{any} convergent subsequence has the same limit at time zero, and the limit evolves with the same equation, then the whole family $ \lbrace \Psi_1^\varepsilon \rbrace $ converges to the same limit $ \Psi $, at all times.

\subsection*{Kirchhoff dynamics on graphs}

We have seen that the eigenstates of $ H_{int}^{\varepsilon, \ell} $ which satisfy Equation \ref{eq:SpectralGapCondition} (the spectral gap condition) give in the limit $ \varepsilon \to 0 $ a contribution to the limit wavefunction which is negligible near the vertex. The limit behaviour at the vertex other than Dirichlet conditions is entirely due to the resonant sequence of eigenstates; we must therefore analyze in detail this resonant sequence. In particular, we are interested in finding how the resonant sequence determines the parameters which characterize the self-adjoint extension of the Laplacian on the graph.

\subsubsection*{The resonant wavefunction $\psi_{k+1}^{\varepsilon, \ell}$}

Let the complex numbers $ \beta_1^{\varepsilon, \ell} \ldots \beta_n^{\varepsilon, \ell} $ be defined by
\[ \psi_{k+1}^{\varepsilon, \ell} \upharpoonright_{ \left\lbrace \ell \right\rbrace \times \Sigma_j^\varepsilon } = \beta_j^{\varepsilon, \ell} \cdot \chi_1^\varepsilon + r_j^{\varepsilon, \ell} \]
where the term $r_j$ stands for the sum of the higher modes $\chi_m^\varepsilon$, $m \geq 2$. We will choose the normalization of $\psi_{k+1}^{\varepsilon, \ell}$ so that
\[ \sum_{j=1}^n | \beta_j^{\varepsilon, \ell} |^2 = 1 \; . \]
With this normalization the numbers $ \beta_j^{\varepsilon, \ell} $ depend on $\varepsilon$ and $\ell$ only through the ratio $\varepsilon/\ell = 1/L$: this is a consequence of the homotheticity of the family of waveguides $ \lbrace \Gamma^\varepsilon \rbrace $. Define
\[ \beta_j = \lim_{ \varepsilon \to 0 } \beta_j^{\varepsilon, \ell } \qquad j = 1 \ldots n \; ; \]
(we assume that the limit exists). Note that by continuity
\[ \sum_{j=1}^n | \beta_j |^2 = 1 \; . \]

Consider the resonant eigenvalue $ \lambda_{k+1}^{\varepsilon, \ell} $. Set $ \sigma = \varepsilon/\ell $; $ \sigma = 1 / L $, so that in our setting $ \sigma \ll 1 $. By a simple scaling argument,
\[ \lambda_{k+1}^{\varepsilon, \ell} = \varepsilon^{-2} \lambda_{k+1}^{1, \ell/\varepsilon} \equiv \varepsilon^{-2} \lambda_{k+1} (\sigma) \]
where $\lambda_{k+1} (\sigma)$ is the $(k+1)$-th eigenvalue of the rescaled mesoscopic region of thickness $ \varepsilon = 1 $ and length $ L = 1 / \sigma $. Define
\[ \theta = \lim_{ \sigma \to 0 } - \frac{\lambda_{k+1} (\sigma) - \lambda_{k+1} (0)}{\sigma} \]
(we assume that this limit exists, too). $ \theta $ is (minus) the derivative of the resonant eigenvalue $ \lambda_{k+1} $ with respect to the parameter $ \sigma $. We point out that $ \theta $ is nonnegative.

The following estimates are useful to recover the boundary conditions at the vertex. Consider the function $ \psi^J $ defined on $ \Gamma_{int}^{\varepsilon, \ell} $, solution of the inhomogeneous Dirichlet problem
\[ \left\lbrace \begin{array}{l}
- \Delta u - \mu_1^\varepsilon u = 0 \\
u \upharpoonright_{\left\lbrace \ell \right\rbrace \times \Sigma^\varepsilon_j} = \beta_j \cdot \chi_1^\varepsilon \quad j=1 \dots n
\end{array} \right. \; . \]

\begin{theorem}\label{thm:FormulaForTheta}
The following formula holds true:
\[ \theta = \lim_{ \varepsilon \to 0 } \varepsilon \cdot \theta^{\varepsilon, \ell} \qquad \theta^{\varepsilon, \ell} \equiv - \oint_{ \partial \Gamma_{int}^{\varepsilon, \ell} } (\psi_{k+1}^{\varepsilon, \ell})^* \partial_\nu \psi^J \; . \]
\theoremend
\end{theorem}

\begin{proof}
Observe that $ \psi_{k+1}^{\varepsilon, \ell} $ is the solution to an inhomogeneous problem similar to the one solved by $ \psi^J $, but with different boundary data on $ \lbrace \ell \rbrace \times \Sigma_j^\varepsilon $:
\[ f_j \equiv \psi_{k+1}^{\varepsilon, \ell} \upharpoonright_{ \lbrace \ell \rbrace \times \Sigma_j^\varepsilon } = \beta_j^{\varepsilon, \ell} \cdot \chi_1^\varepsilon + \ldots \]
Since both the boundary data and the eigenvalues of the two systems converge, we expect the solutions to converge to each other in the limit $ \varepsilon \to 0$ in a suitable Sobolev topology. If the boundary data converge in $ L^2 $, $ \psi^J $ converges to $ \psi_{k+1}^{\varepsilon, \ell} $ uniformly, but their derivatives (in particular, the normal derivatives at the boundary) in general do not.

Consider the Gauss-Green identity
\[ \int_\Omega \Delta f \cdot g - \int_\Omega f \cdot \Delta g = \oint_{\partial \Omega} \partial_\nu f \cdot g - \oint_{\partial \Omega} f \cdot \partial_\nu g \]
and apply it for $ f = ( \psi^J )^* $, $ g = \psi_{k+1}^{\varepsilon, \ell} $, $ \Omega = \Gamma_{int} $; using the equations
\[ - \Delta \psi^J = \mu_1^\varepsilon \cdot \psi^J, \quad - \Delta \psi_{k+1}^{\varepsilon, \ell} = ( \mu_1^\varepsilon + \lambda_{k+1}^{\varepsilon, \ell} ) \cdot \psi_{k+1}^{\varepsilon, \ell} \]
we get
\[ \int_\Omega \Delta f \cdot g - \int_\Omega f \cdot \Delta g = \lambda_{k+1}^{\varepsilon, \ell} \int_{\Gamma_{int}} ( \psi^J )^* \psi_{k+1}^{\varepsilon, \ell} \]
which is equal to
\[ \oint_{\partial \Omega} \partial_\nu f \cdot g - \oint_{\partial \Omega} f \cdot \partial_\nu g \]
\[ = \oint_{\partial \Gamma_{int}} \partial_\nu (\psi^J)^* \cdot \psi_{k+1}^{\varepsilon, \ell} - \oint_{\partial \Gamma_{int}} ( \psi^J )^* \cdot \partial_\nu \psi_{k+1}^{\varepsilon, \ell} = \theta^{\varepsilon, \ell} \; . \]

Note that
\[ \lim_{ \varepsilon \to 0 } \int_{\Gamma_{int}} ( \psi^J )^* \psi_{k+1}^{\varepsilon, \ell} = \lim_{ \varepsilon \to 0 } \int_{\Gamma_{int}} | \psi^J |^2 \]
and
\[ \lim_{ \varepsilon \to 0 } \frac{\int_{\Gamma_{int}} | \psi^J |^2}{\ell} = \sum_{j=1}^n | \beta_j |^2 = 1 \; . \]
Thus
\[ \lim_{ \varepsilon \to 0 } \varepsilon \cdot \theta^{\varepsilon, \ell} = \lim_{ \varepsilon \to 0 } - \varepsilon \cdot \ell \cdot \left( \lambda_{k+1}^{\varepsilon, \ell} \right) \]
\[ = \lim_{ \varepsilon \to 0 } - \frac{\ell}{\varepsilon} \left( \lambda_{k+1}^{1, \ell/\varepsilon} - \lambda_{k+1}^{1, \infty} \right) = \theta \; . \]
(note that $ \lambda_{k+1}^{1, \infty} = 0 $, this term was introduced to emphasize the fact that the definition of $ \theta $ involves ``derivative'' of the resonant energy level $ \lambda_{k+1}^{\varepsilon, \ell} $.
\end{proof}

\subsubsection*{Kirchhoff graphs}

The argument runs exactly in the same way as in the Dirichlet case, except for establishing the vertex conditions, which is what we are going to do. Let us consider a generic wavefunction $ \psi^\varepsilon $ in $ P_{ (0, E] } L^2 (\Gamma^\varepsilon) $; recall that
\[ \psi^\varepsilon \upharpoonright_{ \Gamma_{out}^{\varepsilon, \ell} } = \Psi_1^\varepsilon \otimes \chi_1^\varepsilon + R_1^\varepsilon \; ; \]
by evaluating this expression on $ \lbrace \ell \rbrace \times \Sigma_j^\varepsilon $ we see that the right hand side equals
\[ \Psi_1^\varepsilon (\ell) \cdot \chi_1^\varepsilon + \ldots \]
(dots stand for the negligible higher modes). By Theorem \ref{thm:RelaxOnResonantSequence} we see that
\[ \psi^\varepsilon \upharpoonright_{ \lbrace \ell \rbrace \times \Sigma_j^\varepsilon } = c^\varepsilon_\Psi \cdot \psi_{k+1}^{\varepsilon, \ell} \upharpoonright_{ \lbrace \ell \rbrace \times \Sigma_j^\varepsilon } + \ldots \]
and since
\[ \psi_{k+1}^{\varepsilon, \ell} \upharpoonright_{ \lbrace \ell \rbrace \times \Sigma_j^\varepsilon } = \beta_j^{\varepsilon, \ell} \cdot \chi_1^\varepsilon + \ldots \]
the comparison of the coefficients of the fundamental mode $ \chi_1^\varepsilon $ gives
\[ \Psi_{1, j}^\varepsilon (\ell) = c^\varepsilon_\Psi \cdot \beta_j^{\varepsilon, \ell} + \ldots \]
which, for the limit wavefunction $ \Psi $, reads
\[ \Psi_j (0) = \beta_j \cdot c_\Psi \; . \]

Now consider the restriction of $ \psi^\varepsilon $ to the mesoscopic region $ \Gamma_{int}^{\varepsilon, \ell} $. We make use of the Gauss-Green identity:
\[ \int_\Omega \Delta f \cdot g - \int_\Omega f \cdot \Delta g = \oint_{\partial \Omega} \partial_\nu f \cdot g - \oint_{\partial \Omega} f \cdot \partial_\nu g \; , \]
substituting $\Gamma_{int}$ for $\Omega$, $(\psi^J)^*$ for $f$ and $ \psi^\varepsilon $ for $g$.

The volume integrals vanishes in the limit $ \varepsilon \to 0 $. In fact, using that $ -\Delta \psi^J = \mu_1^\varepsilon \psi^J $ we obtain
\[ \int_\Omega \Delta f \cdot g - \int_\Omega f \cdot \Delta g = \int_{\Gamma_{int}} \left( \Delta \psi^J \right)^* \cdot \psi^\varepsilon - \int_{\Gamma_{int}} \left( \psi^J \right)^* \cdot \Delta \psi^\varepsilon \]
\[ = \int_{\Gamma_{int}} \left( \psi^J \right)^* [ -\Delta - \mu_1^\varepsilon ] \psi^\varepsilon \; . \]
It is clear that $ [ -\Delta - \mu_1^\varepsilon ] \psi^\varepsilon $ belongs to $ P_{(0, E]} L^2 (\Gamma^\varepsilon) $ as $ \psi^\varepsilon $ does: by dividing the integral into the region $ \Omega^\varepsilon $ and the cylinders, it is easy to see that this whole volume term is vanishing like $ \mathcal{O} (\ell) $, or faster.

Now consider the surface integrals. By Theorems \ref{thm:HigherModesOfPsiAtTheBoundary} and \ref{thm:RelaxOnResonantSequence}, we know that at the boundary of $ \partial \Gamma_{int}$ $\psi^\varepsilon $ is proportional to $ \psi_{k+1}^{\varepsilon, \ell} $, plus negligible (in $ L^2 $-norm) terms:
\[ \psi^\varepsilon \simeq c^\varepsilon_\Psi \cdot \psi_{k+1}^{\varepsilon, \ell} \; . \]
Moreover
\[ \oint_{\partial \Omega} \partial_\nu f \cdot g - \oint_{\partial \Omega} f \cdot \partial_\nu g = \oint_{\partial \Gamma_{int}} \left( \partial_\nu \psi^J \right)^* \cdot \psi^\varepsilon - \oint_{\partial \Gamma_{int}} \left( \psi^J \right)^* \cdot \partial_\nu \psi^\varepsilon \]
\[ \simeq c^\varepsilon_\Psi \cdot \oint_{\partial \Gamma_{int}} \left( \partial_\nu \psi^J \right)^* \psi_{k+1}^{\varepsilon, \ell} - \oint_{\partial \Gamma_{int}} \left( \psi^J \right)^* \cdot \partial_\nu \psi^\varepsilon \; . \]

We can substitute the Fourier expansion for $ \psi^\varepsilon $, and as usual the only nontrivial contribution comes from the fundamental mode $ \Psi_1^\varepsilon $: apart from negligible contributions, the surface integral equals
\[ - \theta^{\varepsilon, \ell} \cdot c^\varepsilon_\Psi - \sum_{j=1}^n (\beta_j^{\varepsilon, \ell})^* \cdot (\Psi_{1, j}^\varepsilon)' (0) = \]
\[ - \frac{1}{\varepsilon} \cdot \left( \varepsilon \theta^{\varepsilon, \ell} \right) \cdot c^\varepsilon_\Psi - \sum_{j=1}^n (\beta_j^{\varepsilon, \ell})^* \cdot (\Psi_{1, j}^\varepsilon)' (0) \; . \]

\begin{remark}
The $ \varepsilon^{-1} $ factor above is rather surprising: notice that we can factor it out if we use the microscopic coordinate $ X $ ($ x = \varepsilon X $), so that another $ \varepsilon^{-1} $ comes out of the derivative term. We want to stress that, despite that the use of $ X $ rather than $ x $ may look curious, our aim is, given $ \Gamma^\varepsilon $ ($ \varepsilon $ is small, but finite), to determine \textit{uniquely} $ n+1 $ numbers $ \beta_1 \ldots \beta_n $, $ \theta $ (or $ \theta/\varepsilon $!) which are sufficient to completely describe the effective dynamics on the limit graph $ \Gamma $. The fact that these numbers may be normalized differently, or be covariant with respect to reparametrizations of $ \Psi_1^\varepsilon $, is physically irrelevant.
\end{remark}

If we choose to work with the macroscopic coordinate $X$, we can easily take the limit of the last equation for $ \varepsilon \to 0 $, and we obtain
\[ \sum_{j=1}^n \beta_j^* \cdot \Psi_j' (0) + \theta \cdot c_\Psi = 0 \; . \]
To get quantitative estimates of the error terms, we have to assume some concrete estimate on the convergence of the parameters $ \beta_j^{\varepsilon, \ell} $, $ \theta^{\varepsilon, \ell} $ to their limits: for example,
\[ \beta_j^{\varepsilon, \ell} - \beta_j = \mathcal{O} (L^{-\kappa}), \quad j = 1 \ldots n; \qquad \varepsilon \cdot \theta^{\varepsilon, \ell} - \theta = \mathcal{O} (L^{-\kappa}) \]
for some exponent $ \kappa $ gives an error $ \mathcal{O} (\varepsilon^{-1} L^{-\kappa}) $. Summarizing, we established that the limit wavefunction $ \Psi $ satisfies
\[ \left\lbrace \begin{array}{l} \Psi_j (0) = \beta_j \cdot c_\Psi \\ \sum_{j=1}^n \beta_j^* \cdot \Psi_j' (0) + \theta \cdot c_\Psi = 0 \end{array} \right.
\]
which is sufficient to define the domain of the self-adjoint Hamiltonian $ H_\Gamma $.

\section{(Weak) convergence of resolvents}

In this section we compare the resolvents of the waveguide Hamiltonian $ H^\varepsilon = - \Delta - \mu_1^\varepsilon $ (at finite $ \varepsilon $) and the graph Hamiltonian $ H_\Gamma $. We will discuss explicitly the Kirchhoff case, so we will consider the situation of waveguides satisfying the resonant state condition; the same argument, with straightforward changes, applies to waveguides with spectral gap.

We consider a lift $ H_\Gamma^\varepsilon $ (constructed below) of the self-adjoint extension $ H_\Gamma $ of the free Laplacian on $ \Gamma \backslash \lbrace V \rbrace $. We will prove that the resolvent $ (H^\varepsilon - z)^{-1} $ weakly converges to $ (H_\Gamma^\varepsilon - z)^{-1} $ if and only if the parameters which define $ H_\Gamma $ coincide with those determined by the resonant sequence. Convergence holds true for finite-energy states orthogonal to the bound states of $ \Gamma^\varepsilon $.

Notice that the set of self-adjoint extensions of the free Laplacian on the graph is much larger than the set of extensions that we obtain: our construction indicates that only this small sub-class can be obtained as a limit of the Dirichlet Laplacian on $ \Gamma^\varepsilon $. We shall come back to this point in the conclusions.

According to our construction, the Hamiltonian $ H_\Gamma^\varepsilon $ is zero on states which have no component in the first transversal mode of the channels. Since these states belong to the continuous spectrum of $ H^\varepsilon $, no stronger convergence is to be expected for the resolvents.

\subsection*{Lifting graph Hamiltonians to $\Gamma^\varepsilon$}

Let us consider the self-adjoint Hamiltonian $ H_\Gamma $ on the graph $ \Gamma $, extending the free Laplacian on the edges; the domain of the Hamiltonian is characterized by the set of numbers $ \beta_1 \ldots \beta_n $, $ \theta $ in the following way
\[ \Psi \in \mathcal{D} (H_\Gamma) \iff \left\lbrace \begin{array}{l} \Psi_j (0) = \beta_j \cdot c_\Psi \\ \sum_{j=1}^n \beta_j^* \cdot \Psi_j' (0) + \theta \cdot c_\Psi = 0 \end{array} \right. \]
Given a function $ \Psi : \Gamma \to \mathbb{C} $, in the domain $\mathcal{D}(H)$, we can lift it to a function on $ \Gamma^\varepsilon $ with the operator $ J $ as follows:
\[ \left\lbrace \begin{array}{l} J \Psi \upharpoonright_{ \Gamma_{out} } = \Psi \otimes \chi_1^\varepsilon \\ J \Psi \upharpoonright_{ \Gamma_{int} } = 0 \end{array} \right. \]
In this way, the whole domain $ \mathcal{D} $ can be lifted to $ \Gamma^\varepsilon $: on the closure of this Hilbert subspace, we may define the lifted Hamiltonian $ H_\Gamma^\varepsilon $:
\[ H_\Gamma^\varepsilon (J \Psi) \equiv J ( H_\Gamma \Psi ) \]
and on the complement of $ J \mathcal{D} $ we simply set $ H_\Gamma^\varepsilon $ equal to zero. This means that, for instance, $ H_\Gamma^\varepsilon $ is null on all ``high-energy'' states $ \Psi \otimes \chi_m^\varepsilon $, $ m \geq 2 $.

\subsection*{A useful computation}

The following computation will be needed in the comparison of resolvents; for easier reference, we prove it separately and state the result at the end of the paragraph. Let $ \varphi $ belong to the domain of $ H^\varepsilon $ and $ \psi $ belong to the domain of $ H_\Gamma^\varepsilon $, that is to say, $ \psi = J \Psi $, for $ \Psi \in \mathcal{D}(H_\Gamma) $. We want to compute the difference
\[ \langle H^\varepsilon \varphi, \psi \rangle - \langle \varphi, H_\Gamma^\varepsilon \psi \rangle \; . \]

Let us compute with $ \langle H^\varepsilon \varphi, \psi \rangle $. Since $ \psi $ vanishes on $ \Gamma_{int}^{\varepsilon, \ell} $ this reduces to an integral on $ \Gamma_{out}^{\varepsilon, \ell} $:
\[ \langle H^\varepsilon \varphi, \psi \rangle = \int_{ \Gamma_{out}^{\varepsilon, \ell} } \left( - \Delta \varphi^* \psi - \mu_1^\varepsilon \varphi^* \psi \right) \; . \]
A simple integration by parts leads to
\[ \langle H^\varepsilon \varphi, \psi \rangle = - \oint_{ \partial \Gamma_{out}^{\varepsilon, \ell} } \partial_\nu \varphi^* \cdot \psi + \int_{ \Gamma_{out}^{\varepsilon, \ell} } \left( \nabla \varphi^* \nabla \psi - \mu_1^\varepsilon \varphi^* \psi \right) \; . \]
We write $ \varphi = \Phi \otimes \chi_1^\varepsilon + \ldots $ (dots stand for higher transverse modes): then
\[ - \oint_{ \partial \Gamma_{out}^{\varepsilon, \ell} } \partial_\nu \varphi^* \cdot \psi = c_\Psi \cdot \left( \sum_{j=1}^n \beta_j^* \cdot \Phi_j' (0) \right)^* \; . \]
Now, we compute $ \langle \varphi, H_\Gamma^\varepsilon \psi \rangle $. Again, this is an integral on $ \Gamma_{out}^{\varepsilon, \ell} $ only:
\[ \langle \varphi, H_\Gamma^\varepsilon \psi \rangle = \int_{ \Gamma_{out}^{\varepsilon, \ell} } \varphi^* \cdot \left( - \Psi'' \otimes \chi_1^\varepsilon \right) \; . \]
Substituting the Fourier decomposition of $ \varphi $ and integrating by parts, we get
\[ \langle \varphi, H_\Gamma^\varepsilon \psi \rangle = \sum_{j=1}^n \Phi_j^* (0) \cdot \Psi_j' (0) + \int_\Gamma ( \Phi' )^* \Psi' \; . \]
Recall that $ \Phi_j (0) = c_\Phi \cdot \beta_j^{\varepsilon, \ell} = c_\Phi \cdot \beta_j + \mathcal{O} (\varepsilon^{-1} L^{-\kappa}) $; the sum above can be simplified as follows
\[ \sum_{j=1}^n \Phi_j^* (0) \cdot \Psi_j' (0) = c_\Phi^* \cdot \left( \sum_{j=1}^n \beta_j^* (0) \cdot \Psi_j' (0) \right) + \mathcal{O} (\varepsilon^{-1} L^{-\kappa}) \]
\[ = - \theta c_\Phi^* c_\Psi + \mathcal{O} (\varepsilon^{-1} L^{-\kappa}) \; . \]

It is immediate to verify that $ \int_\Gamma ( \Phi' )^* \Psi' $ equals $ \int_{ \Gamma_{out}^{\varepsilon, \ell} } \left( \nabla \varphi^* \nabla \psi - \mu_1^\varepsilon \varphi^* \psi \right) $; so we conclude that
\[ \langle H^\varepsilon \varphi, \psi \rangle - \langle \varphi, H_\Gamma^\varepsilon \psi \rangle = c_\Psi \cdot \left\lbrace \sum_{j=1}^n \beta_j^* \cdot \Phi_j' (0) + \theta \cdot c_\Phi \right\rbrace^* + \mathcal{O} (\varepsilon^{-1} L^{-\kappa}) \; ; \]
the term in curly brackets vanishes (like $ \mathcal{O} (\ell^{1/2}) + \mathcal{O} (\varepsilon^{-1} L^{-\kappa}) $) if and only if the boundary conditions of $ H_\Gamma $ at the vertex are chosen as those implied by the resonant sequence.

\subsection*{The resolvents of $H^\varepsilon$ and $H_\Gamma^\varepsilon$}

We are now in position to compute the difference of the resolvents of $ H^\varepsilon $ and $ H_\Gamma^\varepsilon $. These operators are bounded, so that it suffices to control their difference on a suitable dense subset of the Hilbert space $ L^2 (\Gamma^\varepsilon) $.

Let $ z $ be a complex number with $ \mathfrak{Im}(z) \neq 0 $; let $ E $ be a fixed positive number. Let us consider $ \varphi, \, \psi \in W^{1, 2} (\Gamma^\varepsilon) $, normalized, with $ \psi $ in the domain of $ H_\Gamma^\varepsilon $. Both $ \varphi $ and $ \psi $ have energy less than $ E $, and are orthogonal to the bound states of $ \Gamma^\varepsilon $ (i.e. they have components in the continuous spectrum only). We prove the following

\begin{theorem}\label{thm:DifferenceOfResolvents}
Consider the expectation value
\[ F^\varepsilon (\varphi, \psi) \equiv \left\langle \varphi, \left( \frac{1}{H^\varepsilon - z} - \frac{1}{H_\Gamma^\varepsilon - z} \right) \psi \right\rangle \; ; \]
where $ \varphi $, $ \psi $ are normalized states in $ L^2 (\Gamma^\varepsilon) $, with $ \varphi $ orthogonal to the bound states and of finite energy: then
\[ F^\varepsilon = \mathcal{O} (\ell^{1/2}) + \mathcal{O} (\varepsilon^{-1} L^{-\kappa}) \; . \]
\theoremend
\end{theorem}

We remark that the constants multiplying the various powers of $ \varepsilon $ and $ \ell $ actually depend on the energy $ E $, and blow up when $ E \to \infty $.

\begin{proof}
We will start with the formula for the difference of the resolvents:
\[ \frac{1}{H^\varepsilon - z} - \frac{1}{H_\Gamma^\varepsilon - z} = \frac{1}{H^\varepsilon - z} (H_\Gamma^\varepsilon - H^\varepsilon) \frac{1}{H_\Gamma^\varepsilon - z} \; . \]
by substituting this formula into our expectation value, we get
\[ F^\varepsilon (\varphi, \psi) = \left\langle \varphi_z, (H_\Gamma^\varepsilon - H^\varepsilon) \psi_z \right\rangle \]
where
\[ \varphi_z = (H^\varepsilon - z)^{-1} \varphi \quad \psi_z = (H_\Gamma^\varepsilon - z)^{-1} \psi \; . \]
We remind that in the computation of the resolvent formula, $H_\Gamma^\varepsilon$ is meant to be applied to $\psi_z$, while $H^\varepsilon$ is applied to $\varphi_z$:
\[ \left\langle \varphi, \left( \frac{1}{H^\varepsilon - z} - \frac{1}{H_\Gamma^\varepsilon - z} \right) \psi \right\rangle = \left\langle \varphi_z, H_\Gamma^\varepsilon \psi_z \right\rangle - \left\langle H^\varepsilon \varphi_z, \psi_z \right\rangle \; . \]

We will compute this formula for $ \varphi \in \mathcal{D} (H^\varepsilon) $ and for $ \psi = J \Psi $, $ \Psi \in \mathcal{D} (H_\Gamma) $. We will briefly sketch how to treat the case of $ \psi $ supported in $ \Gamma_{int}^{\varepsilon, \ell} $ and $ \psi $ supported in $ \Gamma_{out}^{\varepsilon, \ell} $ with high-energy components: the theorem then follows by density.

Notice that $ \psi_z = J \Psi_z $, where $ \Psi_z = (H_\Gamma - z)^{-1} \Psi $. Clearly, $ \Psi_z \in \mathcal{D} (H_\Gamma) $. Moreover, $ \varphi_z \in \mathcal{D} (H^\varepsilon) $. Thus, using the formula that we deduced in the previous paragraph,
\[ \left\langle \varphi_z, H_\Gamma^\varepsilon \psi_z \right\rangle - \left\langle H^\varepsilon \varphi_z, \psi_z \right\rangle = \mathcal{O} (\ell^{1/2}) + \mathcal{O} (\varepsilon^{-1} L^{-\kappa}) \; . \]

If $ \psi $ is compactly supported in $ \Gamma_{int}^{\varepsilon, \ell} $, the action of $ H_\Gamma^\varepsilon $ on $ \psi $ is trivial. The remaining integral on $ \Gamma_{int}^{\varepsilon, \ell} $ is estimated as usual, and is negligible. If $ \psi $ is an high-energy state, the action of $ H_\Gamma^\varepsilon $ is again trivial, and we notice that the components of $ \varphi $ along the higher transverse modes are suppressed: hence, the integral vanishes as $ \varepsilon \to 0 $ and $ \ell \to 0 $.

\end{proof}

This theorem implies that on scattering states with finite energy the difference of the resolvents goes to zero if $ \varepsilon \to 0 $ and $ \ell \to 0 $.

\section*{The scattering problem}

Consider a waveguide $ \Gamma^\varepsilon $ in $ \mathbb{R}^3 $, and a smooth bounded fast decaying potential $ V $ defined on $ \mathbb{R}^3 $. Let as usual denote by $ H^\varepsilon $ the operator $ - \Delta - \mu_1^\varepsilon $, where $ \Delta $ is the Laplacian with Dirichlet boundary conditions on $ \partial \Gamma^\varepsilon $. Assume that $ \Gamma^\varepsilon $ admits exactly one resonant sequence and let $ H_\Gamma = - \Delta_\Gamma $, where $ \Delta_\Gamma $ is the Laplacian on $ \Gamma $ with vertex conditions adapted to the resonant sequence. Both $ H^\varepsilon $ and $ H_\Gamma $ have the positive real axis as absolutely continuous spectrum. In addition, $ H^\varepsilon $ may have bound states. One can define the wave operators
\[ W^\pm (H^\varepsilon + V, H^\varepsilon) = s - \lim_{t \to \mp \infty} e^{it(H^\varepsilon + V)} e^{-itH^\varepsilon} P_{a.c.} \]
where $ P_{a.c.} $ is the projection onto the absolutely continuous spectrum, and
\[ W^\pm (H_\Gamma + V_\Gamma, H_\Gamma) = s - \lim_{t \to \mp \infty} e^{it(H_\Gamma + V_\Gamma)} e^{-itH_\Gamma} \]
where $ V_\Gamma $ is the evaluation of $ V $ on the graph $ \Gamma $. Asymptotic completeness holds in both cases, and one can define unitary S-matrices $ S^\varepsilon (V) $, $ S_\Gamma ( V ) $. It is easy to find the explicit form of $ W^\pm (H_\Gamma + V_\Gamma, H_\Gamma) $ with perturbative methods of stationary scattering theory, since the generalized eigenfunctions of $ H_\Gamma $ are known. It is more difficult to give an explicit form of $ W^\pm (H^\varepsilon + V, H^\varepsilon) $; therefore it is useful to have an approximation of it in terms of $ W^\pm (H_\Gamma + V_\Gamma, H_\Gamma) $, with error estimates. We sketch here the way this approximation is obtained by the method of resonant sequences: it is essentially an adaptation of the chain rule. Details will be given elsewhere.

We start with the identity, valid for all $ t $ and all positive values of the small parameter $ \varepsilon $:
\[ e^{-itH_\Gamma^\varepsilon} e^{it(H_\Gamma^\varepsilon + V)} e^{-it(H_\Gamma^\varepsilon + V)} e^{it(H^\varepsilon + V)} = e^{-itH_\Gamma^\varepsilon} e^{itH^\varepsilon} e^{-itH^\varepsilon} e^{it(H^\varepsilon + V)} \]
(clearly, we use the lift $ H_\Gamma^\varepsilon $ of $ H_\Gamma $; the wave operators for $ H_\Gamma^\varepsilon $ are the obvious extension of the wave operators of $ H_\Gamma $). By Wiener's lemma, for every $ E > 0 $ one has
\[ \lim_{\varepsilon \to 0} \lim_{T \to \infty} \frac{1}{T} e^{-itH_\Gamma^\varepsilon} e^{itH^\varepsilon} \Pi_E = \Pi_E \quad \lim_{\varepsilon \to 0} \lim_{T \to \infty} \frac{1}{T} e^{-it(H_\Gamma^\varepsilon + V)} e^{it(H^\varepsilon + V)} \Pi_E = \Pi_E \]
where $ \Pi_E $ is the projection operator on states in which $ H^\varepsilon < E $. Since strong limits can be composed one has for any $ E $
\[ \lim_{\varepsilon \to 0} \Pi_E W^\pm (H^\varepsilon + V, H^\varepsilon) \Pi_E = \Pi_E \lim_{\varepsilon \to 0} W^\pm (H_\Gamma^\varepsilon + V, H_\Gamma^\varepsilon) \Pi_E \; . \]
It is easy to recognize in the last term $ W^\pm (H_\Gamma + V_\Gamma, H_\Gamma) $, after the natural identifications.

\section{Comparison with previous results}

We compare briefly our approach and results with previous work on the same problem. There are many published papers on the subject, here we limit ourselves to only some of them.

As we remarked, our main emphasis is on the fact that the limit dynamics at a vertex of the graph depends on the spectral properties of a sequence of Schr{\"o}dinger operators defined on neighborhoods of the vertex in the fattened graph, and that in order to obtain in the limit boundary conditions other than Dirichlet it is necessary that a resonating sequence exists for this sequence of operators.

The role of some sort of resonance has been stressed by B. Pavlov (see \cite{pavlov07a} and previous papers). In \cite{albeverio07} this is rigorously proved in the case of a graph shrinking to a line with a sharp bending, in the context of limit point interactions. In this context the asymptotes at $ \pm \infty $ of a zero-energy resonance provide, in the suitable scaling limit, a connection between the two sides of the singular point on the line. This role of the resonating sequence also in the case of graphs has been our guiding idea.

In a very interesting paper \cite{post05} O. Post has proved (among other things) that some shrinking of one of the tubes towards the vertex to which it is attached is enough to provide Dirichlet boundary conditions in the limit for the edge corresponding to that tube. Our results imply that the shrinking is not necessary if a resonant sequence does not exist.

In another relevant paper S. Molchanov and B. Vainberg \cite{molchanov07} prove (among other things) continuity of the scattering matrix under the shrinking operation $ \varepsilon \to 0 $ for ``fattened graphs''. This is also sufficient to establish a weak form of resolvent convergence, very similar to the one that we propose. This abstract result does not give information on the role of the geometry of a neighborhood of the vertex and therefore does not use information that comes from the specific geometries suggested by physical examples, e.g. from the density of conduction electrons on aromatic molecules.

A refined analysis of the structure of eigenfunctions can be found in the very interesting and relevant paper of D. Grieser \cite{grieser08a}, which again exploits the strategy of studying the scattering matrix. The paper of Grieser is very detailed and full of relevant results. In particular one can extract from it that in the case of ``fattened graphs'' which shrink uniformly in $ \varepsilon $ a resonant sequence has energy exponentially close to the continuum treshold. While our approach is local and relies on energy esitmates and Neumann-Robin bracketing, the approach of \cite{grieser08a} relies on connecting smoothly the eigenfunction of the internal region with scattering waves. In spite of the different estimates, there are points in common between the two approaches. It would in particular be interesting to extend Grieser's results to sets $ \Gamma^\epsilon $ with nonuniform scaling, which have a geometry more in line with physical data.

Note that a suitable notion of resolvent convergence can be achieved both within our approach, and through the continuity in $ \varepsilon $ of scattering matrices. We think that our approach makes the role of the resonant state very transparent: we introduce the resonant sequence, which is a kind of ``local'' (restricted to a neighbourhood of the junction) notion of resonance, and we identify the parameters of the limit graph Hamiltonian in terms of the asymptotic values and the energy of the resonant sequence. Moreover, the dependence of the limit Hamiltonian on the geometry of the junction is actually through the spectrum of the mesoscopic region, which is a significant simplification.

In our paper we give an estimate of the error one makes as a function of $ \varepsilon $ on the evaluation of relevant physical quantities when using the limit dynamics; this can be done in all approaches which take into account the geometric shape of the neighborhoods.

\section{Conclusions and generalizations}

Let $ \Gamma^\varepsilon $ be an $ \varepsilon $-thin, star-shaped waveguide and consider a Schr{\"o}dinger evolution generated by the Hamiltonian $ H^\varepsilon = - \Delta - \mu_1^\varepsilon $, where $ \Delta $ is the Dirichlet Laplacian on $ \Gamma^\varepsilon $; choose $ \mu_1^\varepsilon $ so that the onset of the continuous spectrum is $0$. Define a sequence of auxiliary (disconnected) waveguides $ \Gamma_{out}^{\varepsilon, \ell} $, isometric to $ \Gamma \times \Sigma^\varepsilon $, together with a sequence of lifted Hamiltonians $ H_\Gamma^\varepsilon $; under the assumption that there exist (at most one) resonant sequence of eigenstates, we prove that the resolvent of $ H_\Gamma^\varepsilon $ converges weakly to the resolvent of $ H^\varepsilon $, on low-energy states. Estimates of the difference are given, as well as estimates for the approximation of observables such as density and (longitudinal) current.

The argument is constructive: the parameters that identify $ H_\Gamma $ are uniquely determined in terms of simple properties of the resonant sequence, which in turn depends on the geometry of the vertex. When the same results hold for free Hamiltonians perturbed by an external potential $ V $, the wave operators $ W^\pm (H^\varepsilon + V, H^\varepsilon) $, ``projected'' on low-energy states, indeed reduce to their analogues $ W^\pm (H_\Gamma + V_\Gamma, H_\Gamma) $ on the graph $ \Gamma $, and this provides a useful approximation formula. If the resonant sequence is absent, we prove that the limit Quantum Graph is the decoupled one. The case with more than one resonant sequence is not addressed here.

We point out that the effective dynamics on metric graphs that we obtain with our approximation method belongs to a subset of all possible self-adjoint extensions of the free Laplacian on $ \Gamma \backslash \lbrace V \rbrace $ that are characterized by the boundary conditions
\[ \Psi_j (0) = \beta_j \cdot c_\Psi, \qquad \sum_{j=1}^n \beta_j^* \cdot \Psi_j' (0) + \theta \cdot c_\Psi = 0 \]
where $ \beta_j $ is the asymptotic value of the resonance in the $ j $-th branch of the waveguide and $ \theta $ is the derivative of the resonant energy. In view of applying our methods to study diffusion on thin waveguides \cite{freidlin93}, we remark that these, whenever the $ \beta $'s can be chosen nonnegative, are precisely the conditions under which $ H_\Gamma $ generates a positivity preserving contraction semigroup on $ \Gamma $ (with continuous trajectories when $ \beta_j = \beta_k $) \cite{kostrykin08}. If $ \beta_{k_0} = 0 $ for some $k_0$ the corresponding branch carries an independent process. On the other hand if  $ H^\varepsilon $ has no bound states it generates a positivity preserving contraction semigroup on $ \Gamma^\varepsilon ;$  it is natural to expect  that  the stochastic process generated by $ H^\varepsilon $ when averaged in the transversal direction converges to the process generated by $ H_\Gamma $ when $ \varepsilon \to 0 $. We shall come back to this issue in a future publication.

We conclude with a few remarks. The result described here holds also in the case in which the arms of the waveguide are $ \varepsilon $-neighbourhoods of smooth curves, provided that their sections $ \Sigma^\varepsilon $ are compact, of linear size in $ \varepsilon $ and  the bottom of their Dirichlet spectrum is nondegenerate.  The method can also be used, with simple modifications, to treat generic waveguides which are $ \varepsilon $-neighbourhoods of graphs with a locally finite number of vertices (better: the distance between any couple of vertices is bounded from below by a positive constant) \cite{costa10}.

Even if the method described here is used for a sequence of homotetic thin waveguides, it can be applied to the case where  the shape of the $ \varepsilon $-waveguide depends on the parameter $ \varepsilon $. In particular, it can be applied to study a two-dimensional strip with a bend which becomes sharp as $ \varepsilon \to 0 $. In \cite{albeverio07} this problem was studied in two dimensions, the curved strip was shrunk to a smooth curve and \textit{then} the curve was sharply bent. In this case the motion on the curve is described by the free Schr{\"o}dinger equation plus an attractive potential (due to the curvature) that in the limit $ \varepsilon \to 0 $ gets deeper and concentrates around the origin (the sharp bending point). If the parameters are conveniently chosen one can obtain a zero-energy resonance and a limit motion characterized by Kirchhoff-like boundary conditions, with parameters depending on the resonance. For all other choices one gets decoupling conditions at the origin.

As for concrete examples of resonant sequence, taking inspiration from  pictures of aromatic molecules, we will show that one can have a resonant sequence for a waveguide composed of a thick spherical shell $ R_\varepsilon \leq |x| \leq R'_\varepsilon $ attached  smoothly to cylinders of base $\Sigma_\varepsilon $ where the parameters are chosen in such a way that the first eigenvalue of the Dirichlet Laplacian in the spherical shell is equal to the first eigenvalue of the Dirichlet Laplacian on $ \Sigma_\varepsilon$.

\section*{Aknowledgemts}

We are pleased to thank here S. Albeverio and C. Cacciapuoti for useful discussions in the early stages of this research. One of us (GFDA) is grateful to C. Amovilli for showing him pictures of the density of conduction electrons in aromatic molecules. We are also grateful to the anonymous referee for attracting our attention to refs \cite{molchanov07}, \cite{grieser08} and for useful suggestions.

\section*{References}

\bibliographystyle{jphysicsB}
\bibliographystyle{plain}

\bibliography{quantum_graphs}

\end{document}